\documentclass{article}

\usepackage{authblk}
\usepackage{graphicx}
\usepackage{wrapfig}
\usepackage{braket}
\usepackage{bbold}
\usepackage{amsmath}
\usepackage{amssymb}
\usepackage{amsthm}
\usepackage{dsfont}
\usepackage{tikz}
\usepackage{quantikz}
\usepackage{amsfonts,amsxtra}
\usepackage[
    backend=biber,
    bibstyle=ieee,
    citestyle=numeric-comp,
    sorting=nyt,
]{biblatex}
\usepackage[
  top=2.5cm,
  bottom=2.5cm,
  left=2.5cm,
  right=2.5cm
]{geometry}
\usepackage{xcolor}
\usepackage[most]{tcolorbox}
\usepackage{hyperref}
\hypersetup{
    colorlinks=true,
    linkcolor=blue,
    citecolor=blue,
    urlcolor=blue
    }
\usepackage{subcaption}
\newsavebox{\bccaption}

\newtheorem{theorem}{Theorem}
\newcommand{\SEP}{\operatorname{SEP}}

\title{Quantum Heuristic Separation of Comb Inequalities for the Travelling Salesman Problem}
\author[1]{Yohan Finet\thanks{yohan.finet@usherbrooke.ca}}
\author[1]{Victor Drouin-Touchette\thanks{victor.drouin-touchette@usherbrooke.ca}}
\affil[1]{Institut Quantique and Département de Génie Électrique et de Génie Informatique, Faculté de Génie, Université de Sherbrooke, Sherbrooke, Québec, J1K 2R1, Canada}

\begin{document}

\maketitle

\begin{abstract}
Comb inequalities are essential cutting planes for the travelling salesman problem, yet their separation remains computationally challenging. We propose a quantum heuristic separation algorithm for this family of inequalities. The algorithm uses a procedure that prepares a uniform superposition of all combs with a fixed number of teeth using polynomially many gates and qubits. As part of this procedure, we present a circuit that coherently encodes the Hamming weight of a data register in one-hot form in an auxiliary register. The separation algorithm then searches for violated inequalities using the Grover-Mixer Quantum Alternating Operator Ansatz. Its running time is polynomial under prescribed polynomial budgets for the number of QAOA rounds, parameter optimization steps and samples gathered, although we provide no evidence that such budgets suffice to find a violated comb with constant probability when one exists. We provide an explicit circuit construction for the mixer unitary, the full problem Hamiltonian and a resource analysis.
\end{abstract}

\section{Introduction}
The Travelling Salesman Problem (TSP) asks for a minimum-weight Hamiltonian cycle in a weighted complete graph. When the distance between two cities is the same in both directions, it is called a symmetric TSP; otherwise, it is called an asymmetric TSP. The latter can be transformed into a symmetric problem by doubling the number of vertices. TSP appears in many fields such as logistics (planning school bus routes or postal distribution), biology (genome sequencing), electronics (testing integrated circuits and drilling boards) and astronomy (planning telescopic observations) among others \cite{applegate2007}. The decision version of the TSP was shown by Karp to be NP-complete \cite{Karp1972}. This result also holds for metric and Euclidean instances \cite{Papadimitriou1977}. Furthermore, it is NP-hard to obtain an approximate solution to metric instances with an approximation factor less than $\frac{123}{122}$ in the symmetric case and $\frac{75}{74}$ in the asymmetric case \cite{Karpinski2015}. For some $c>10^{-36}$, Karlin et al.~\cite{Karlin2021} gave a polynomial-time $\frac{3}{2}-c$ approximation algorithm for symmetric metric instances. On the other hand, Arora showed a polynomial-time approximation scheme for Euclidean instances \cite{arora1996}.

The most successful exact approaches to the TSP employ a branch-and-cut framework. The efficiency of this search depends critically on the cutting planes (the inequalities) added to tighten the linear programming (LP) relaxation used in this algorithm. The task of finding, within a given class, inequalities violated by a solution to the LP relaxation is called separation. One important class for the TSP is that of comb inequalities, introduced in their general form by Gr\"{o}tschel and Padberg \cite{Grotschel1979a,Grotschel1979b}. The separation problem for this class was recently shown to be NP-hard by the authors in a companion paper~\cite{finet2026}.

Quantum computing offers alternative strategies for combinatorial optimization problems. Grover's search algorithm \cite{grover1996} achieves a quadratic speedup for unstructured search. Brassard et al.\ \cite{brassard2002} generalised this to amplitude amplification, which boosts the probability of measuring a marked state from $p$ to near unity using $O(1/\sqrt{p})$ iterations. When the fraction of marked states is exponentially small in the problem size, the required number of iterations grows exponentially, limiting the utility of amplitude amplification as a standalone algorithm. Optimization by decoded quantum interferometry (DQI) \cite{Jordan2025} also uses interference to approximate solutions to optimization problems. However, the algorithm requires specifying the problem in a particular way which makes it difficult to adapt it to a specific problem.

A wide range of combinatorial optimization problems can be reformulated as quadratic unconstrained binary optimization (QUBO) problems \cite{Lucas2014}, in which constraints are incorporated into the objective function through slack variables and penalty terms scaled by Lagrange multipliers. Two major quantum heuristic paradigms are well-suited to this formulation. Rooted in the adiabatic theorem \cite{born1928}, quantum annealing \cite{Kadowaki1998} prepares a system in the ground state of an initial Hamiltonian and slowly evolves it toward a problem Hamiltonian whose ground state encodes the optimal solution. The quantum approximate optimization algorithm (QAOA) \cite{farhi2014} is a hybrid classical-quantum technique in which a quantum state is constructed through alternating applications of problem and mixer unitaries, whose parameters are optimized by a classical outer loop. Each approach carries notable practical drawbacks: quantum annealing is constrained by hardware connectivity that frequently conflicts with the graph structure of the QUBO instance, whereas deeper QAOA circuits can accumulate more errors on noisy hardware and make classical parameter optimization more demanding.

We believe that achieving near-term quantum advantage in real-world settings will demand hybrid algorithms that tightly couple quantum subroutines with classical solvers, judiciously exploiting the strengths of current quantum devices where they are most beneficial. With that in mind, we introduce in this paper a quantum heuristic separation algorithm for comb inequalities and show how to integrate it into a hybrid optimizer for the TSP. Given a fixed fractional solution $x^* \in [0,1]^{|E|}$ (where the number of edges $|E| = n(n-1)/2$ for the complete graph) obtained from the LP relaxation, we show how to prepare, for a given number of teeth, a quantum superposition of all corresponding comb structures in a quantum register. In doing so, we also prove Theorem~\ref{thm:main}, where $\mathcal{C}_s^n$ is the set of combs with $s$ labelled teeth on a graph with $n$ vertices, with different orderings of the teeth counted as distinct elements.

\begin{theorem}\label{thm:main}
For every odd integer $s$ satisfying $3 \leq s \leq \lfloor n/2 \rfloor$, the uniform state
\begin{equation}
    \ket{F_s^n} = \frac{1}{\sqrt{|\mathcal{C}_s^n|}}\sum_{C \in \mathcal{C}_s^n} \ket{C},
\end{equation}
can be prepared exactly using $O(sn)$ qubits and $O(sn^3)$ one- and two-qubit gates when arbitrary single-qubit rotations are available.

Alternatively, for every $0<\epsilon<1$, a circuit over the finite gate set $\{H,S,T,\mathrm{CNOT}\}$ can prepare the state $\ket{F_s^n}$ to trace-distance error at most $\epsilon$ using $O(ns)$ qubits and $O(sn^3\log(n/\epsilon))$ gates.
\end{theorem}

Consequently, quantum amplitude amplification can be used to search for violated combs, while the Grover-Mixer Quantum Alternating Operator Ansatz (GM-QAOA) offers a heuristic alternative with a potentially lower circuit cost per iteration. We present how the GM-QAOA algorithm can be adapted to obtain violated combs by providing an explicit circuit construction for the mixer unitary and the full problem Hamiltonian. The running time of this algorithm is polynomial under prescribed polynomial budgets for the number of QAOA rounds, parameter optimization and sampling, although we provide no evidence that these budgets suffice to find a violated comb with constant probability when one exists. An investigation of performance for realistic instances is left for future work.

\section{Background}

\subsection{Linear programming}

A linear programming (LP) problem aims to optimize a linear objective function under linear constraints. Its standard form is given by:
\begin{equation}\label{eq:linear_programming}
  \begin{split}
    \text{maximize} \quad & c^\intercal x \\
    \text{subject to} \quad & Ax \leq b, \quad x \in \mathbb{R}_{\geq 0}^n.
  \end{split}
\end{equation}
The feasibility problem for linear programs is known to be P-complete under logarithmic-space reductions~\cite{Dobkin1979, Khachiyan1979}. George Dantzig introduced in 1947 the simplex algorithm \cite{dantzig1990} to solve these problems. The algorithm is very efficient in practice~\cite{Borgwardt1982,Spielman2003}, but has exponential complexity in the worst case \cite{klee1972}. Here we will use $x$ to indicate a generic vector variable and $x^*$ for a fractional point.

Let $G = (V, E)$ be the complete undirected graph with $n = |V|$ vertices and edge weights $c_{ij}$ for each edge $(i,j)$. The TSP can be formulated as an integer linear program (ILP) by introducing binary variables $x_{ij} \in \{0,1\}$ indicating whether edge $(i,j)$ is included in the tour:

\begin{subequations}\label{eq:TSP_ILP}
\begin{alignat}{3}
& \text{minimise} \quad
&& \sum_{i < j} c_{ij} x_{ij} \qquad
&& \label{eq:TSP_ILP:line1}\\
& \text{subject to} \quad
&& \sum_{j \ne i} x_{ij} = 2 \qquad
&& \forall i \in V \label{eq:TSP_ILP:line2}\\
&
&& \sum_{i,j \in S, i < j} x_{ij} \leq |S| - 1 \qquad
&& \forall S \subset V \: : \: 2 \leq |S| \leq n - 1
   \label{eq:TSP_ILP:line3}\\
&
&& x_{ij} \geq 0 \qquad
&& \forall (i,j) \in E \label{eq:TSP_ILP:line4}\\
&
&& x_{ij} \in \{0, 1\} \qquad
&& \forall (i,j) \in E, \label{eq:TSP_ILP:line5}
\end{alignat}
\end{subequations}
where the degree constraints \eqref{eq:TSP_ILP:line2} ensure that every vertex is incident to exactly two tour edges, and the subtour elimination constraints \eqref{eq:TSP_ILP:line3} prevent the formation of disconnected cycles over subsets of~$V$. The convex hull in $\mathbb{R}^{|E|}$ of incidence vectors $x$ satisfying constraints \eqref{eq:TSP_ILP:line2} -- \eqref{eq:TSP_ILP:line5} is called the Symmetric Travelling Salesman Polytope and is denoted by STSP($n$). The polytope obtained when removing \eqref{eq:TSP_ILP:line5} is called the Subtour Elimination Polytope and is denoted by SEP($n$).

Since ILP optimization is NP-hard in general \cite{garey1979}, solving the TSP formulation~\eqref{eq:TSP_ILP} directly is intractable for large instances. The most successful exact approaches follow a work by Dantzig, Fulkerson and Johnson who proposed a cutting-plane method \cite{Dantzig1954} for the TSP. First, an LP relaxation of \eqref{eq:TSP_ILP} is formulated where only constraints \eqref{eq:TSP_ILP:line2} and \eqref{eq:TSP_ILP:line4} are retained. The method then iteratively solves the LP relaxation using the simplex and adds violated constraints taking the form \eqref{eq:TSP_ILP:line3} or other forms to the linear program to tighten the LP relaxation. Note that this procedure will always provide a lower bound on the value of the objective function for the optimal solution to the TSP. However, the ability to converge to an integer solution will depend on the ability of the procedure to find violated inequalities. Many classes of inequalities have been proven to be facet defining for STSP($n$), meaning that the points of the polytope satisfying each such inequality with equality form a facet of this polytope. For example, Gr\"{o}tschel \& Padberg \cite{Grotschel1979a,Grotschel1979b} proved that both subtour inequalities and comb inequalities (discussed in the next subsection) induce facets of STSP($n$). Clique-tree inequalities, a generalization of comb inequalities, were introduced and proven to be faced defining for STSP$(n)$ by Gr\"{o}tschel \& Pulleyblank \cite{Grotschel1986}. Path inequalities, another generalization of comb inequalities introduced by Cornuéjols, Funlupt and Naddef \cite{Cornuéjols1985}, were also proven to be facet defining for STSP$(n)$ by Naddef \& Rinaldi \cite{Naddef2007}. Many other valid or faced-defining inequalities for STSP$(n)$ exist, see for example \cite{Naddef2001} for a more extensive list. These ideas eventually led to the development of the branch-and-cut method which is at the core of the Concorde solver \cite{applegate2003} considered to be the leading exact solver for the TSP.

We now describe the basics of a branch-and-cut algorithm for the TSP illustrated in Fig.~\ref{fig:branch-and-cut} (see for example \cite{Padberg1990b} or \cite{applegate1995} for more details). The method builds a search tree where each node represents a system of linear inequalities. First, an initial solution is obtained using a heuristic such as Christofides' algorithm \cite{christofides1976} or the Lin-Kernighan algorithm \cite{lin1973}. At this point in the algorithm, this initial solution serves as the best solution found so far. The cutting-plane method is then performed as previously described. As more cuts are added, the increases in the optimal value of the LP relaxation tend to become smaller. When these increases become too small or no more cuts can be found, branching is performed. Branching consists of partitioning the set of tours $\mathcal{S}$ into two sets $\mathcal{S}_1$ and $\mathcal{S}_2$. For example, this can be done by setting the value of a decision variable to 0 or 1:
\begin{equation}
  \mathcal{S}_1=\{x \in \mathcal{S}:x_e=0\}, \; \mathcal{S}_2=\{x \in \mathcal{S}:x_e=1\},
\end{equation}
{\parfillskip=0pt
which provides two new subproblems associated with two new child nodes in the branch-and-cut tree, one for each of $\mathcal{S}_1$ and $\mathcal{S}_2$ and the parent node being associated to the LP problem on $\mathcal{S}$. The cutting-plane method is then performed on the first subproblem and then to the other which adds more inequalities to the system\par
}

\clearpage
\noindent
\begin{minipage}[b][\textheight][t]{0.48\textwidth}
    \setlength{\parindent}{1.5em}
    
    \noindent
     of linear inequalities. When the cutting-plane method yields an integer solution that is better than the best integer solution found so far, this new solution is saved as the best integer solution.
    
    Since the optimal solution of the relaxed LP problem provides a lower bound on the values of the integer solutions, if the value of the objective function obtained from the resolution of the relaxed LP associated to a given node in the tree is greater or equal to the value of the best integer solution obtained so far, this node can be pruned knowing that there is no integer solutions with a smaller objective value that can be obtained from this node or any of its descendants. The choice of the next subproblem to be processed depends on the implementation. For example, Hong \cite{Hong1972} used depth-first search whereas the code of Padberg \& Rinaldi \cite{padberg1991} and Concorde \cite{applegate2003} use best-bound search where the subproblem having the lowest LP bound is chosen. The branching process continues until all leaf nodes are pruned or provide an integer solution.

    An important step of the cutting-plane method consists in finding violated inequalities. Given a class of inequalities $\mathcal{C}$ and a vector $x^*$, a separation algorithm for $\mathcal{C}$ either returns an inequality in $\mathcal{C}$ that is violated by $x^*$ or a failure message. A separation algorithm is exact if it only returns a failure message when all inequalities in $\mathcal{C}$ are satisfied by $x^*$ while a heuristic separation algorithm may return a failure message even when some inequality in $\mathcal{C}$ is violated by $x^*$. Although there are exponentially many subtour elimination constraints, a polynomial-time exact separation algorithm for subtour inequalities was shown by Hong \cite{Hong1972} and further improved by Crowder \& Padberg \cite{Crowder1980} and Padberg \& Rinaldi \cite{Padberg1990a}. However, there are no known exact polynomial-time algorithm for the separation of comb inequalities \cite{applegate2007, Cook2019} and a proposed proof of NP-hardness for this problem was recently provided by the authors~\cite{finet2026}.
    
    \subsection{Separation of comb inequalities}
    
    In this work, we are interested in the comb inequalities of Gr\"{o}tschel \& Padberg \cite{Grotschel1979a,Grotschel1979b} who generalized classes of inequalities introduced by Chv{\'a}tal \cite{Chvatal1973} and Edmonds \cite{edmonds1965}. Given a complete undirected graph $G=(V,E)$ and the vector variable $x$ introduced in the previous subsection, the comb inequality is given by
    \begin{equation}\label{eq:comb_ineq}
        x\bigl(\delta(H)\bigr)+\sum_{j=1}^{s}x\bigl(\delta(T_j)\bigr) \geq 3s+1,
    \end{equation}
\end{minipage}\hfill
\begin{minipage}[b][\textheight][t]{0.49\textwidth}

  \sbox{\bccaption}{%
    \begin{minipage}{\linewidth}
      
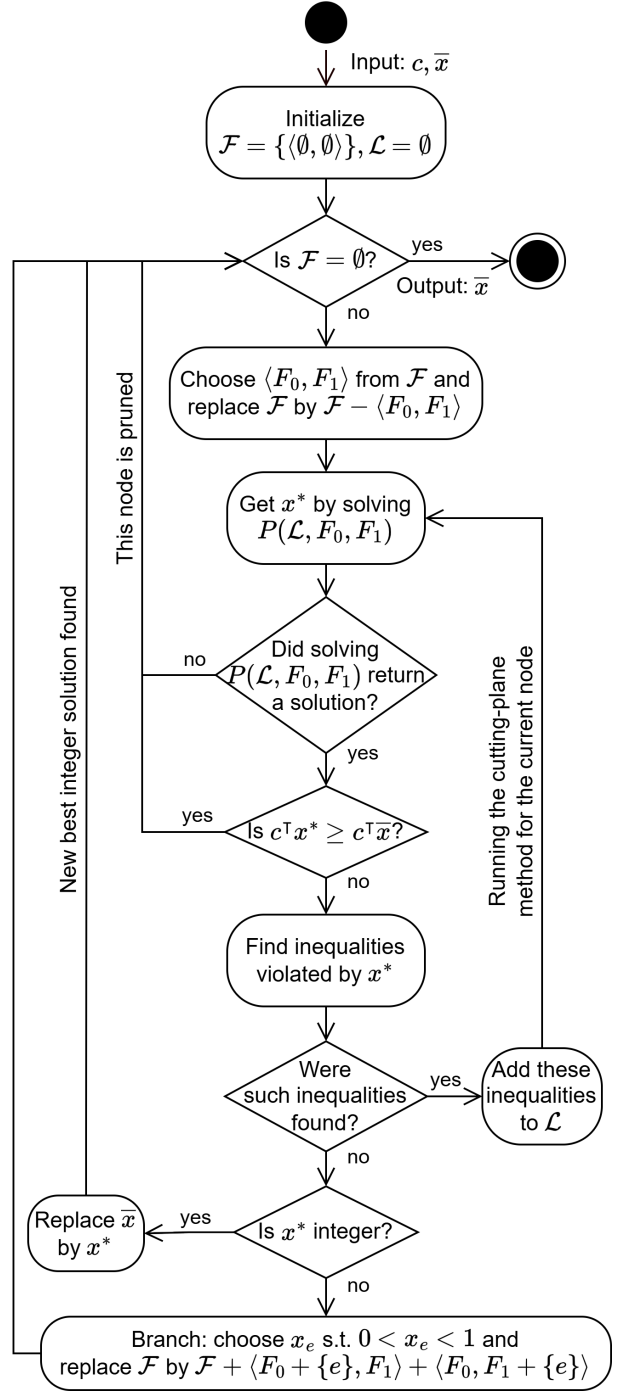
\captionof{figure}{Branch-and-cut algorithm. The cost vector is denoted by $c$ and $\overline{x}$ is the best integer solution found so far. The pair $\langle F_0, F_1 \rangle$ consists of disjoint edge sets $F_0, F_1 \subseteq E$. $\mathcal{L}$ is a set of inequalities, and $P(\mathcal{L}, F_0, F_1)$ is the linear program minimizing $c^\intercal x$ subject to \eqref{eq:TSP_ILP:line2}, \eqref{eq:TSP_ILP:line4}, $\mathcal{L}$, $x_e = 0$ for $e \in F_0$ and $x_e = 1$ for $e \in F_1$. Its solution is denoted $x^*$.}
      \label{fig:branch-and-cut}
    \end{minipage}%
  }

  \centering
  \includegraphics[
    width=\linewidth,
    height=\dimexpr
      \textheight-\ht\bccaption-\dp\bccaption-8pt
    \relax,
    keepaspectratio
  ]{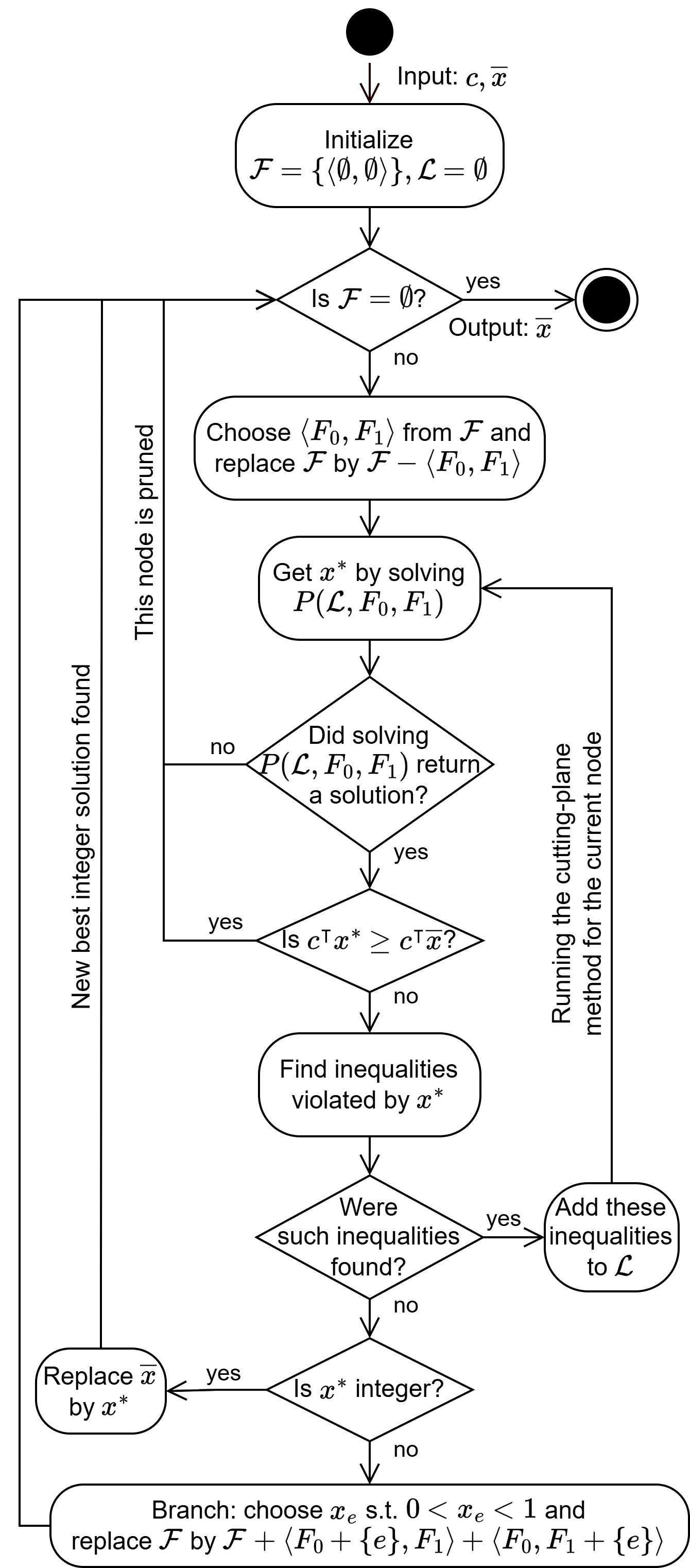}
  \par
  \vspace{8pt}
  \vfill
  \nointerlineskip
  \usebox{\bccaption}
\end{minipage}
\par
\clearpage

\noindent
where the handle $H \subset V$ and teeth $T_1,\ldots,T_s \subset V$ illustrated in Fig.~\ref{fig:comb} satisfy the following conditions:

\begin{enumerate}
    \item \label{cond1}$T_j \setminus H \neq \emptyset$ for all $j$,
    \item \label{cond2}$T_j \cap H \neq \emptyset$ for all $j$,
    \item \label{cond3}$T_i \cap T_j = \emptyset$ for all $i \neq j$,
    \item \label{cond4}$s \geq 3$ and $s$ is odd.
\end{enumerate}
Here, $\delta(S) = \{e \in E : |e \cap S| = 1\}$ denotes the cut induced by $S \subseteq V$, and $x(F) = \sum_{e \in F} x_e$ denotes the sum of the components of the vector $x$ over the edge set $F \subseteq E$.

Despite the NP-hardness of separating comb inequalities~\cite{finet2026}, the operations research community has produced a rich set of approaches. Padberg \& Rao \cite{Padberg1982} described a polynomial-time separation algorithm for the 2-matching inequalities (also known as blossom inequalities), which are comb inequalities having the special property that $|T_j \cap H|=1$ and $|T_j \backslash H|=1$ for all $j$. Their algorithm solves a sequence of $O(|E^*|)$ max-flow problems. Given a fractional solution $x^*$ obtained from the relaxed LP problem, we have $E^*:=\{ e \in E: x_e^* > 0\}$ and the support graph $G^*$ is the subgraph of $G$ induced by $E^*$. Using the pre-flow push algorithm \cite{Goldberg1988}, their algorithm can be implemented to run in $O(n|E^*|^2\log(n^2/|E^*|))$ time. For combs with a fixed number of teeth $s$, Carr~\cite{Carr1997} devised a separation method by solving $O(n^{2s})$ maximum flow problems, where $n$ denotes the number of vertices. However, since $s$ can be any odd integer between $3$ and $\lfloor n/2 \rfloor$, this method does not run in polynomial time when the number of teeth is not bounded. For fractional points of SEP($n$) whose support graph is planar, Fleischer and Tardos~\cite{Fleischer1999} proposed an algorithm with $O(n^2 \log n)$-time complexity to identify maximally violated comb inequalities. A comb is maximally violated if its slack is $-1$ in Eq.~\eqref{eq:comb_ineq}, the largest possible violation when $x^* \in \SEP(n)$. Letchford~\cite{Letchford2000} introduced the broader family of domino-parity inequalities and obtained an $O(n^3)$-time separation algorithm under the same planarity assumption. Cook, Espinoza, and Goycoolea~\cite{Cook2010} subsequently extended this result to a larger family, a subset of which includes all clique-tree inequalities and admits polynomial-time separation when the support graph is planar and the number of handles is bounded by a fixed constant. Caprara, Fischetti and Letchford~\cite{Caprara2000} demonstrated that comb inequalities belong to the family of $\{0, \frac{1}{2}\}$-cuts and provided an $O(n^2|E^*|)$-time algorithm to detect maximally violated ones. Fleischer, Letchford and Lodi~\cite{Fleischer2006} established the feasibility of polynomial-time separation for a superclass of simple comb inequalities when $x \in \text{SEP}(n)$. Note that these algorithms for larger families may return violated inequalities that are not comb inequalities. Consequently, the results presented by Caprara, Fischetti and Letchford, for example, do not automatically extend to comb separation under the same assumptions since their algorithm may return a maximally violated \{0,$\frac{1}{2}$\}-cut that is not a comb inequality whether a maximally violated comb inequality exists or not. Heuristic separation methods were also proposed by Padberg and Rinaldi~\cite{Padberg1990b}, Clochard and Naddef~\cite{Clochard1993}, and Naddef and Thienel~\cite{Naddef2002}.

\begin{figure}
  \centering
  \includegraphics[width=0.4\textwidth]{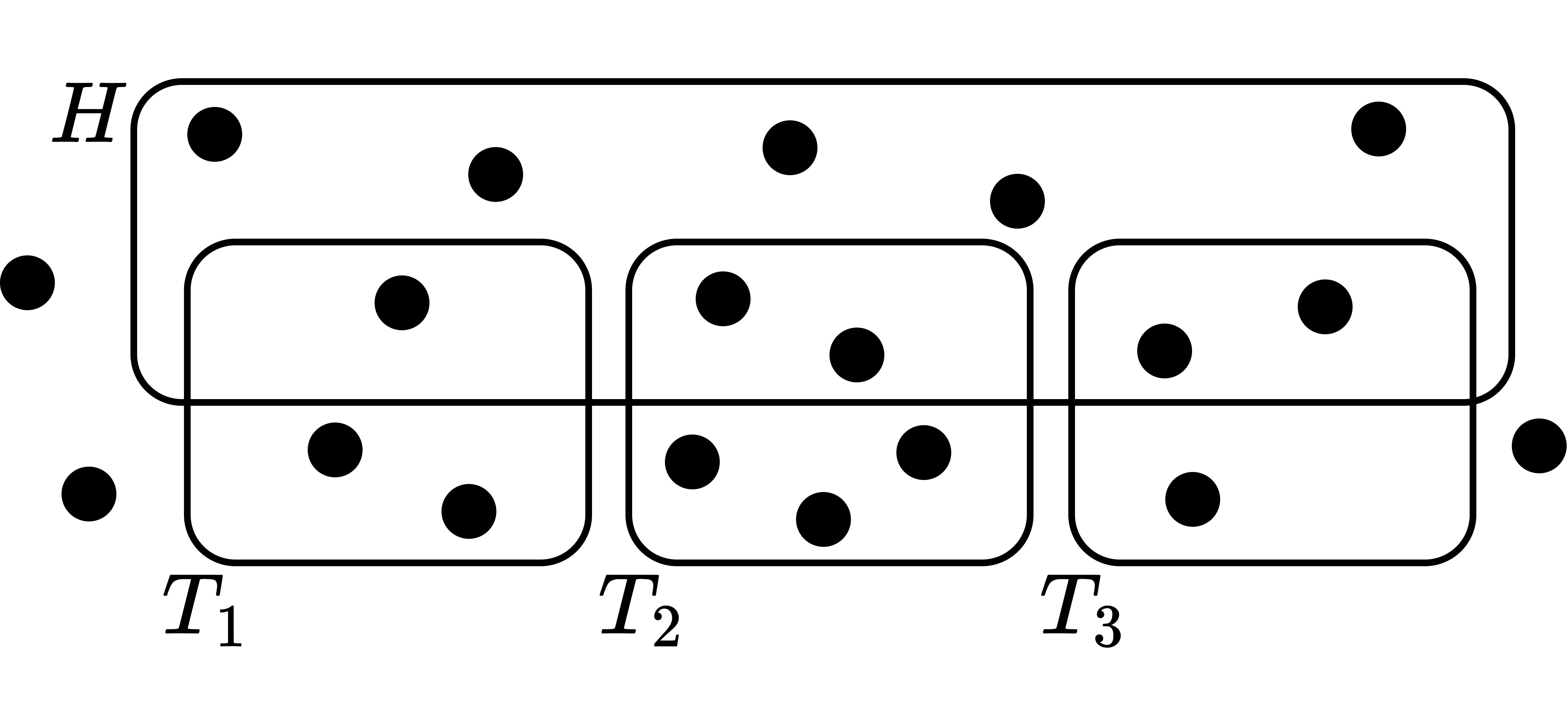}
  \caption{A comb with handle $H$ and teeth $T_1$, $T_2$ and $T_3$.}
  \label{fig:comb}
\end{figure}

\subsection{The Quantum Alternating Operator Ansatz}

One of the prominent quantum heuristic approaches to combinatorial optimization is based on the quantum approximate optimization algorithm (QAOA) \cite{farhi2014}. In this algorithm, given an objective function $f$ to be optimized, a phase Hamiltonian $H_P$ diagonal in the computational basis encodes $f$ such that ${H_P \ket{x}=f(x)\ket{x}}$. The mixing Hamiltonian $H_M$ is the transverse-field Hamiltonian $H_M=\sum_{i=1}^n \sigma_i^x$, where $n$ is the number of qubits and $\sigma_i^x$ denotes the Pauli $\sigma_x$ operator acting on qubit $i$. The quantum computer is initialized in the uniform superposition
\begin{equation}
    \ket{s}= \frac{1}{\sqrt{2^n}}\sum_x \ket{x},
\end{equation}
which is the ground state of $-H_M$. A parameterized state is prepared through $p$ rounds of alternating evolution under $H_P$ and $H_M$:
\begin{equation}\label{eq:qaoa}
    \ket{\boldsymbol{\beta},\boldsymbol{\gamma}}=e^{-i\beta_p H_M}e^{-i\gamma_p H_P} \dots e^{-i\beta_2 H_M}e^{-i\gamma_2 H_P}e^{-i\beta_1 H_M}e^{-i\gamma_1 H_P}\ket{s},
\end{equation}
where $\boldsymbol{\beta},\boldsymbol{\gamma}\in\mathbb{R}^p$. For each parameter choice, the state in Eq.~\eqref{eq:qaoa} is repeatedly prepared and measured in the computational basis. The sampled objective values provide an estimate of the expectation $\langle f \rangle=\bra{\boldsymbol{\beta},\boldsymbol{\gamma}} H_P \ket{\boldsymbol{\beta},\boldsymbol{\gamma}}$. A classical optimization loop updates the parameters to optimize this expectation, and the algorithm returns the best solution sampled.

Hadfield et al.~\cite{Hadfield2019} generalized this framework through the quantum alternating operator ansatz (also QAOA), allowing parameterized unitary families that need not arise from evolution under fixed Hamiltonians. For constrained optimization over a feasible set $F$, the computational-basis states encoding feasible solutions span a subspace of the Hilbert space spanned by the qubits to which the circuit evolution can be restricted. Equation~\eqref{eq:qaoa} becomes
\begin{equation}
    \ket{\boldsymbol{\beta},\boldsymbol{\gamma}}=U_M(\beta_p)U_P(\gamma_p) \dots U_M(\beta_2)U_P(\gamma_2)U_M(\beta_1)U_P(\gamma_1)\ket{s}.
\end{equation}
Here, $\ket{s}$ is a feasible initial state, typically preparable by a constant-depth circuit and often chosen to encode a single feasible solution. The phase-separation unitaries $U_P(\gamma)$ depend on the objective function, while the mixing unitaries $U_M(\beta)$ depend on the feasible domain and its structure. The phase separator must be diagonal in the computational basis and is usually chosen as $U_P(\gamma)=e^{-i\gamma H_P}$. The mixing unitary is required to take feasible states to feasible states and to provide transitions between all pairs of feasible states.

Grover-Mixer QAOA (GM-QAOA)~\cite{Bartschi2020} is a particular realization of the quantum alternating operator ansatz algorithm that uses $U_M(\beta)=e^{-i\beta \ket{F} \bra{F}}$, where 
\begin{equation}
    \ket{F}=\frac{1}{\sqrt{|F|}}\sum_{x\in F}\ket{x}    
\end{equation}
is the uniform superposition over the feasible set $F$. The initial state is taken to be $\ket{s}=\ket{F}$ and is prepared by a unitary $U_S$ satisfying $U_S\ket{0}^{\otimes n}=\ket{F}$. The mixer can therefore be expressed as
\begin{equation}
    U_M=U_S (\mathds{1}-(1-e^{-i\beta})(\ket{0}\bra{0})^{\otimes n}) U_S^\dagger,
\end{equation}
which can be implemented using the circuits for $U_S$ and $U_S^\dagger$, two layers of $X$ gates and a multi-controlled phase shift gate. This construction shifts the challenge of mixer design to the efficient preparation of the uniform superposition over feasible solutions.

\section{GM-QAOA for comb inequalities}

\subsection{Encoding}

We begin by describing how an arbitrary comb with $s$ teeth can be encoded using qubits. A quantum register of $n$ qubits represents the handle, and an additional register of the same size is allocated for each tooth. Within each register, qubits are indexed from $1$ to $n$, and vertex $i$ is considered to belong to the handle (or to a tooth) if and only if qubit $i$ is in state $\ket{1}$ in the corresponding register. Taken together, these registers form the comb register comprising at most $N_{\mathrm{comb}} = n\left(\lfloor \frac{n}{2} \rfloor + 1 \right)$ qubits in total as there are $\lfloor \frac{n}{2} \rfloor$ possible teeth. Note that permuting the tooth labels produces different computational-basis states that represent the same comb inequality. This symmetry could in principle be exploited to reduce the number of states used to encode the set of valid combs by a factor of up to $\lfloor \frac{n}{2} \rfloor!$, though we do not pursue this reduction in the present work.

We note that a different encoding of combs can be defined. Since teeth are pairwise disjoint, every vertex belongs to exactly one of the following categories:
\begin{itemize}
    \item $H \cap T_j$ for some $j$ (in the handle and in tooth $j$),
    \item $T_j \setminus H$ for some $j$ (in tooth $j$ but not in the handle),
    \item $H \setminus \bigcup_j T_j$ (in the handle but in no tooth),
    \item $V \setminus (H \cup \bigcup_j T_j)$ (in neither the handle nor any tooth).
\end{itemize}
Each vertex therefore falls into one of $2s + 2$ categories: for each tooth $j$, there are two categories ($H \cap T_j$ and $T_j \setminus H$), plus handle only and neither. Combs can thus be represented using $n$ qudits of dimension $2s+2$. Encoding each $d$-dimensional qudit requires $\lceil\log_2(d)\rceil$ qubits, so $n\lceil\log_2(2s+2)\rceil$ qubits would be needed, which is at most $n\lceil\log_2(n+2)\rceil$ since $s\leq\lfloor n/2\rfloor$. This encoding would reduce the number of qubits required and the circuit would have to be adapted, but note that it doesn't change the size of the solution space to be explored since the QAOA mixer keeps states in the feasible subspace. For the remainder of this work, we will use the $N_{\mathrm{comb}}$-qubit encoding discussed previously.

\subsection{Phase-separation operator}

The input to the quantum separation algorithm is a fractional LP solution $x^* \in [0,1]^{|E|}$, obtained by solving the LP relaxation of the TSP, where each component $x_e^*$ is a continuous value representing the fractional weight assigned to edge $e$ in the relaxed solution. We now describe how these fractional values are used to construct a cost Hamiltonian encoding the following objective function:
\begin{equation}\label{eq:obj_func}
    \text{minimise} \quad x^*\bigl(\delta(H)\bigr)+\sum_{j=1}^{s}x^*\bigl(\delta(T_j)\bigr).
\end{equation}
As indicated by~\eqref{eq:comb_ineq}, a violated comb inequality corresponds to an objective value in~\eqref{eq:obj_func} smaller than $3s+1$.

The fractional weights $x_e^*$ are encoded as interaction energy coefficients in the Hamiltonian as follows. For a set $S_i$ (handle or tooth), let $b_{k,i}\in\{0,1\}$ indicate whether vertex $k$ belongs to $S_i$ and let $\ket{C}$ denote the quantum state encoding the comb $C$. For an edge $e = (k,l)$, we have $e \in \delta(S_i)$ if and only if $b_{k,i}\oplus b_{l,i}=1$. Let $\sigma^z_{k,i}$ denote the Pauli-$Z$ operator acting on qubit $k$ of register $i$. The projector onto states in which exactly one endpoint of $e$ belongs to $S_i$ is
\begin{equation}\label{eq:xor_ising}
    \Pi_{e,i}
    = \frac{1}{2}
       \bigl(\mathds{1}-\sigma^z_{k,i}\sigma^z_{l,i}\bigr).
\end{equation}
Indeed, $\Pi_{e,i}\ket{C} = (b_{k,i}\oplus b_{l,i})\ket{C}$.

Equation~\eqref{eq:xor_ising} introduces a $ZZ$ coupling between the corresponding comb register qubits. We assign an energy coefficient $x_e^*$ to each such interaction. Thus, for a given state encoding a particular comb ${C=(H, T_1, \ldots, T_s)}$, the total energy contribution from the cuts is given by the Hamiltonian
\begin{equation}\label{eq:cut_energy}
  \begin{split}
    H_{\mathrm{cuts}}
    &= \sum_{e \in \delta(H)} x_e^*\Pi_{e,H}  + \sum_{j=1}^{s} \; \sum_{e \in \delta(T_j)} x_e^* \Pi_{e,T_j}\\
    &= \sum_{e=\{k,l\} \in \delta(H)} \frac{x_e^*}{2}\bigl(\mathds{1} - \sigma^z_{k,H}\, \sigma^z_{l,H}\bigr) + \sum_{j=1}^{s} \; \sum_{e=\{k,l\} \in \delta(T_j)} \frac{x_e^*}{2}\bigl(\mathds{1} - \sigma^z_{k,T_j}\, \sigma^z_{l,T_j}\bigr)\\
    &=\sum_{i=1}^{s+1} \; \sum_{e=\{k,l\} \in E} \frac{x_e^*}{2}\bigl(\mathds{1} - \sigma^z_{k,i}\, \sigma^z_{l,i}\bigr),
  \end{split}
\end{equation}
where we used the fact that ${\sigma^z_{k,i} \sigma^z_{l,i} \ket{\psi} = \ket{\psi}}$ when the edge $(k,l)$ is not in $\delta(S_i)$ to obtain the sum over $E$ on the last line. For every valid comb $C$, this operator satisfies
\begin{equation}\label{eq:H_cut_eigenvalues}
    H_{\mathrm{cuts}}\ket{C}
    =
    \left[
        x^*\bigl(\delta(H)\bigr)
        +\sum_{j=1}^{s}x^*\bigl(\delta(T_j)\bigr)
    \right]\ket{C}.
\end{equation}
Therefore, its eigenvalue on $\ket{C}$ is precisely the left-hand side of the comb inequality~\eqref{eq:comb_ineq}. The identity terms in Eq.~\eqref{eq:cut_energy} contribute the same constant energy shift to every state, which does not affect which state minimizes the energy and can thus be ignored. We therefore define the cost Hamiltonian as
\begin{equation}\label{eq:combs_cost_hamiltonian}
  H_C = -\sum_{e=\{k,l\} \in E}\; \sum_{i=1}^{s+1} \frac{x_e^*}{2}\, \sigma^z_{k,i}\, \sigma^z_{l,i},
\end{equation}
which is diagonal in the computational basis, as expected. The phase-separation unitary is then defined in the usual way as
\begin{equation}
  U_P(\gamma)=e^{-i\gamma H_C}.
\end{equation}

We note that only a small fraction of the terms appearing in equation \eqref{eq:combs_cost_hamiltonian} are nonzero. Indeed, the fractional solution is obtained by solving a LP relaxation of the following form:
\begin{equation}\label{eq:TSP_LP}
\begin{split}
\text{minimise} \quad & c^T x\\
\text{subject to} \quad
& x(\delta(v)) = 2 \quad \forall v \in V\\
& Cx \leq d\\
& 0 \leq x_e \leq 1 \quad \forall e \in E,
\end{split}
\end{equation}
where $Cx \leq d$ denotes a system of $m$ inequalities satisfied by all feasible tours. It is a property of the simplex method to guarantee that any optimal solution $x^*$ to \eqref{eq:TSP_LP} has at most $2n + m$ nonzero components~\cite{applegate2007}. Consequently, we expect the number of gates required to implement the phase-separation operator to be relatively small.

\subsection{Mixing operator}

We now present a circuit implementing the state-preparation unitary $U_S$ and requires only a polynomial number of gates to prepare a superposition of all the possible combs. This circuit is formed of a tooth-preparation circuit $U_T$ (section \ref{sec:tooth}) and a handle-preparation circuit $U_H$ (section \ref{sec:handle}) such that ${U_S = U_H U_T}$.

We note that, by combining this circuit with an oracle verifying inequality~\eqref{eq:comb_ineq}, one could in principle design a quantum exact separation algorithm via amplitude amplification \cite{brassard2002}. However, as we discuss in Section \ref{sec:conclusion}, the number of combs grows exponentially with the number of vertices so this approach would require an exponential number of Grover iterations if the number of violated combs doesn't grow at the same pace. Therefore, GM-QAOA acts as a more friendly heuristic approach.

\subsubsection{Tooth-preparation circuit} \label{sec:tooth}

The first step consists of preparing a superposition of the possible teeth. The circuit presented in Fig.~\ref{fig:teeth_circuit} achieves this and we now describe it. Since there are $s$ teeth and each must contain at least two vertices, the first tooth can contain from $2$ to $n - 2(s-1)$ vertices. Representing a superposition of all the possible assignments of $k$ vertices to this tooth can be done by encoding the Dicke state $\ket{D_k^n}$~\cite{Dicke1954} in the register corresponding to this tooth. Dicke states are equal superpositions of all $\binom{n}{k}$ $n$-qubit computational-basis states with Hamming weight $k$. The state $\ket{D_k^n}$ can be prepared with a circuit of size $O(nk)$ and depth $O(n)$. It follows that a state of the form $\sum_{k=0}^n e^{i\phi_k}\alpha_k \ket{D_k^n}$ can be prepared with a circuit of size $O(n^2)$ and depth $O(n)$ [Theorem 2 of Ref.~\cite{Bartschi2019}]. We use such a circuit to prepare the state $\sum_{k=2}^{n-2(s-1)} \alpha_k \ket{D_k^n}$ in the register associated with tooth $T_1$. In Figures~\ref{fig:teeth_circuit} and~\ref{fig:handle_circuit}, we use the gate $D_{p,q}^r$ to denote a circuit preparing the superposition
\begin{equation}\label{eq:Dicke_superposition}
  D_{p,q}^r \ket{0}^{\otimes r} = \sum_{k=p}^{q} \alpha_k \ket{D_k^r}.
\end{equation}
We discuss how the coefficients $\alpha_k$ are determined in Section~\ref{sec:uniform_distribution} to ensure that a uniform superposition of combs is prepared by the circuit. For now, we leave these coefficients unspecified. We use a $Mask$ register of $n$ qubits that we initialize in the $\ket{1}^{\otimes n}$ state using an $X$ gate on each qubit of this register. Then, each qubit $i$ of the $Mask$ register acts as the target of a $CNOT$ gate that is performed with qubit $i$ of the $T_1$ register as control. Therefore, for each state of the $T_1$ register, the state of the $Mask$ register indicates which vertices are not included in tooth $T_1$ and can thus be included in the other teeth.

We then introduce a procedure that assigns vertices to tooth $T_j$ based on the assignments to the preceding $j-1$ teeth. This procedure is applied successively to teeth $T_2$ through $T_s$ and the corresponding circuit is illustrated between the two red dashed vertical lines in Fig.~\ref{fig:teeth_circuit}. First, using the $Mask$ register, we encode in a $Count$ register of $n$ qubits the number of vertices that are not yet included in a tooth using a one-hot encoding. Therefore, the state of qubit $c$ of the $Count$ register is $\ket{1}$ if $c$ qubits in the $Mask$ register are in state $\ket{1}$. This is performed by the multi-controlled $E$ gate for which we describe a circuit implementing it in Section~\ref{subsubsec:one-hot_circuit}.

Now, tooth $T_j$ must have at least two vertices assigned to it. In the case where $c$ vertices have not yet been assigned to any tooth, tooth $T_j$ can be assigned at most $c - 2(s-j)$ vertices to keep at least two vertices for each of the remaining $s-j$ teeth. Therefore, for states of the $Mask$ register containing $c$ ones, we would like to prepare a superposition of all assignments of 2 to $c - 2(s-j)$ vertices chosen among the $c$ not yet assigned vertices to tooth $T_j$. To achieve this, we prepare in a $Dicke$ register of $n$ qubits the state $\sum_{k=2}^{c-2(s-j)} \alpha_k \ket{D_k^c}$ with the help of a controlled version of the $D_{p,q}^r$ gate with qubit $c$ of the $Count$ register as control. This is repeated from $c=2(s-j+1)$, which corresponds to the case where only two vertices remain for each of the tooth $T_j$ to $T_s$, to $c=n-2(j-1)$ corresponding to the case where only two vertices where assigned to each of the first $j-1$ teeth. Consequently, this procedure affects only the first $n-2(j-1)$ qubits of the $Dicke$ register.

All that remains is to swap the first $c$ qubits of the $Dicke$ register to the available positions in the $T_j$ register. This is done using $n$ controlled swaps, each exchanging the first qubit of the $Dicke$ register with qubit $i$ of the $T_j$ register and qubit $i$ of the $Mask$ register acting as control. Each of these controlled swap is followed by a controlled cyclic rotation of the first $n-2(j-1)$ qubits of the $Dicke$ register with qubit $i$ of the $Mask$ register as control. Since the $Mask$ qubits are used as controls, exactly $c$ of these operations will be performed for a given state of the $Mask$ register. The controlled cyclic rotation ensures that these $c$ operations will swap out all $c$ qubits of the $Dicke$ register with their clean $T_j$ counterpart. Hence, the $Dicke$ register will have returned to the $\ket{0}^{\otimes n}$ state. Finally, the $Count$ register is uncomputed and the $Mask$ register is updated to keep track of the vertices assigned to tooth $T_j$.

\begin{figure}
  \centering
  \includegraphics[width=\textwidth]{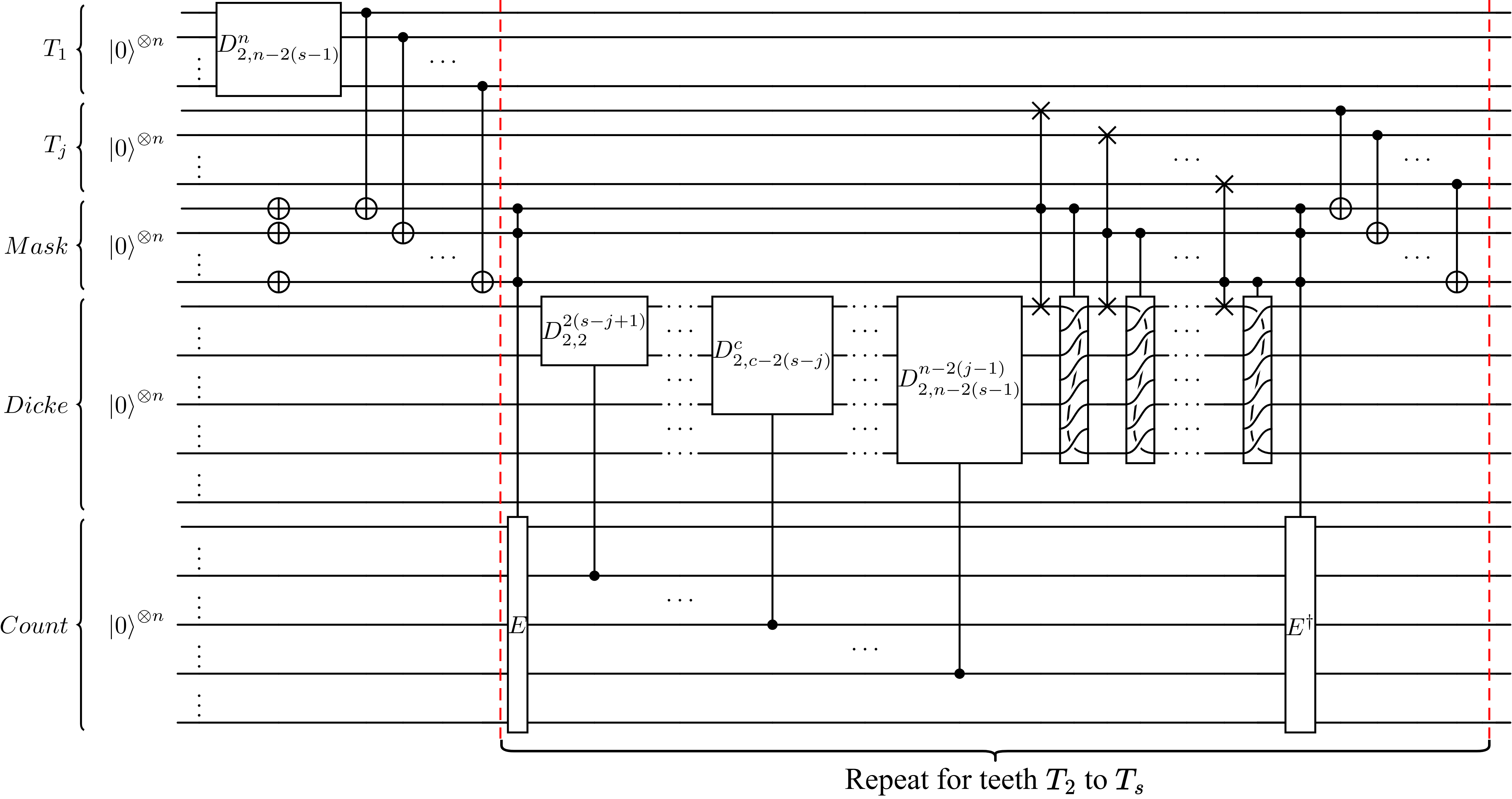}
  \caption{Circuit $U_T$ for preparing a superposition of teeth assignments.}
  \label{fig:teeth_circuit}
\end{figure}

\subsubsection{Handle-preparation circuit} \label{sec:handle}

Having prepared the superposition of teeth, we now prepare a superposition of all compatible vertex assignments to the handle. The circuit used to achieve this is illustrated in Fig.~\ref{fig:handle_circuit}. We first note that for a tooth with $v$ vertices, the handle must contain at least one of them and at most $v-1$. The idea is therefore to create a superposition of all these possibilities.

The handle register is initially in the $\ket{0}^{\otimes n}$ state while the other registers carry their states from tooth-preparation circuit. The following procedure is repeated for each tooth; we describe it for tooth $T_j$. We use the multi-controlled $E$ gate to encode in the $Count$ register the number of vertices assigned to tooth $T_j$. If there are $v$ qubits in state $\ket{1}$ in tooth $T_j$ we create the superposition $\sum_{k=1}^{v-1} \alpha_k \ket{D_{k}^v}$ using the controlled version of the $D_{p,q}^r$ gate with qubit $v$ of the $Count$ register as control. This is repeated for $v=2$ to $v=n-2(s-1)$. As in the tooth-preparation circuit, we now swap the first $v$ qubits in the $Dicke$ register to the available positions in the $H$ register. This is done using $n$ controlled swaps, each exchanging the first qubit of the $Dicke$ register with qubit $i$ of the $H$ register and qubit $i$ of the $T_j$ register acting as control. Each of these controlled swaps is followed by a controlled cyclic rotation of the first $n-2(s-1)$ qubits of the $Dicke$ register with qubit $i$ of the $T_j$ register as control. Since the $T_j$ register qubits are used as controls, exactly $v$ of these operations will be performed for a given state of the $T_j$ register, after which the $Dicke$ register will have returned to the $\ket{0}^{\otimes n}$ state. Finally, the $Count$ register is uncomputed.

At this point in the circuit, the $Mask$ register encodes which vertices are contained in no teeth. Since each vertex in no tooth may either belong to the handle or lie outside it, we perform a controlled Hadamard gate with qubit $i$ of the $Mask$ register as control and qubit $i$ of the $H$ register as target. Finally, the $Mask$ register is uncomputed by performing a CNOT with qubit $i$ of the $Mask$ register as target and qubit $i$ of register $T_j$ as control. This is repeated for all teeth before applying an $X$ gate to all $Mask$ register qubits.

\begin{figure}[t]
  \centering
  \includegraphics[width=\textwidth]{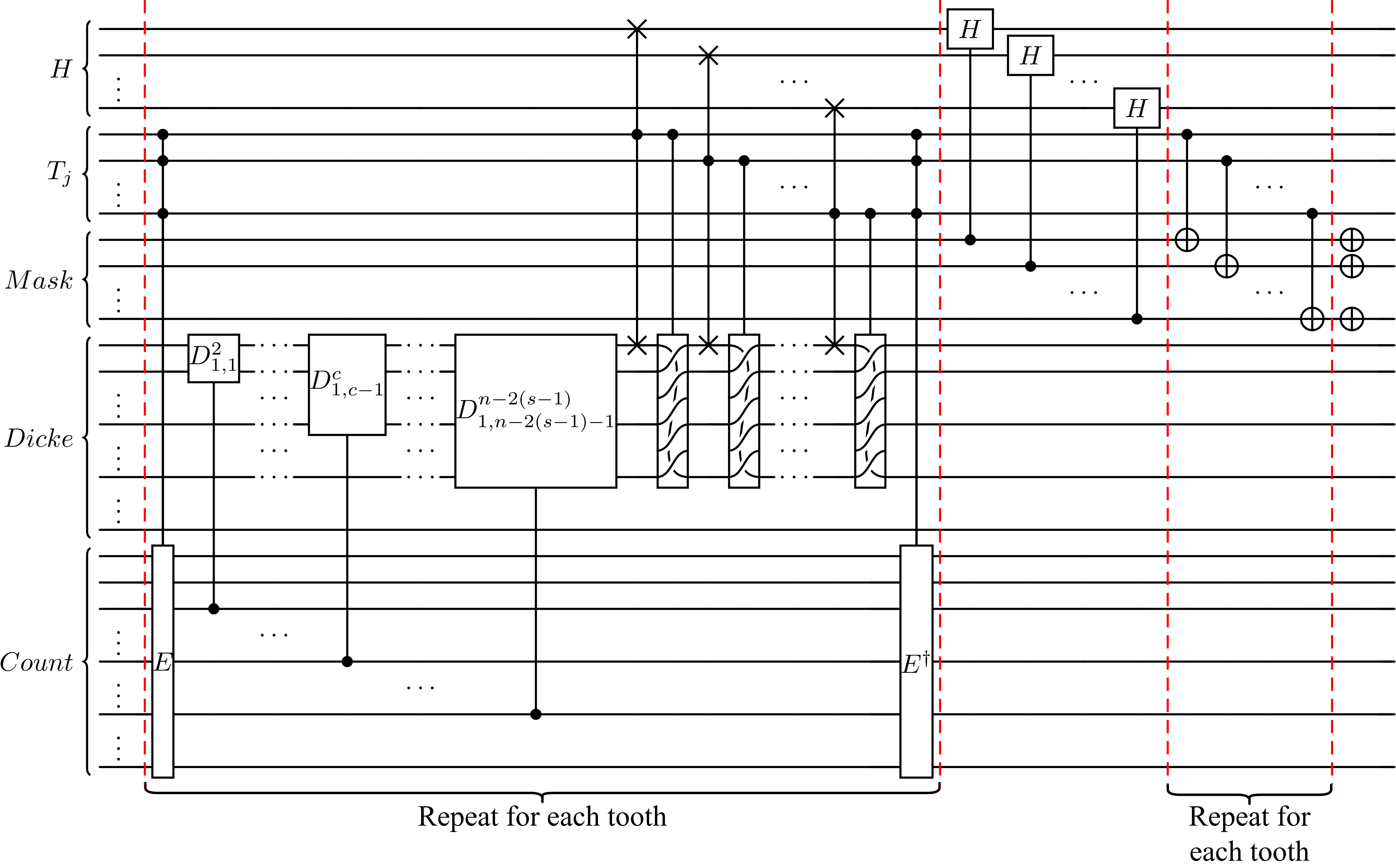}
  \caption{Circuit $U_H$ for preparing a superposition of vertex assignments to the handle. The $T_j$ and $Mask$ registers carry their states from tooth preparation while the other registers are in the $\ket{0}^{\otimes n}$ state initially.}
  \label{fig:handle_circuit}
\end{figure}

\subsubsection{The one-hot encoding circuit}\label{subsubsec:one-hot_circuit}

Figure~\ref{fig:one-hot_circuit} illustrates the implementation of the multi-controlled $E$ gate for $n=3$. This circuit is inspired by a similar one presented by Kaye and Mosca~\cite{Kaye2001}. Initially, qubit $0$ of the $Count$ register is prepared in the state $\ket{1}$, while all other qubits in this register are initialized to $\ket{0}$. A cascade of conditional swap gates (Fredkin gates~\cite{Fredkin1982}) is then applied, with the data qubits serving as controls. These data qubits correspond to the $Mask$ register in Fig.~\ref{fig:teeth_circuit} and to the $T_j$ registers in Fig.~\ref{fig:handle_circuit}. For each data qubit in state $\ket{1}$, the $\ket{1}$ state in the $Count$ register is shifted by one position. Consequently, after processing all data qubits, the Hamming weight of the $Data$ register is encoded in one-hot form in the $Count$ register: qubit $c$ is in the state $\ket{1}$ if and only if exactly $c$ data qubits are in the state $\ket{1}$. Since the Hamming weight of the $Data$ register is guaranteed to be strictly less than $n$ for the circuits presented in Figures~\ref{fig:teeth_circuit} and~\ref{fig:handle_circuit}, an $n$-qubit $Count$ register, indexed from $0$ to $n-1$, is sufficient in these cases. For $n$ data qubits, the circuit uses $O(n^2)$ Fredkin gates, each of which have constant-size decomposition into two-qubit gates~\cite{Smolin1996}.

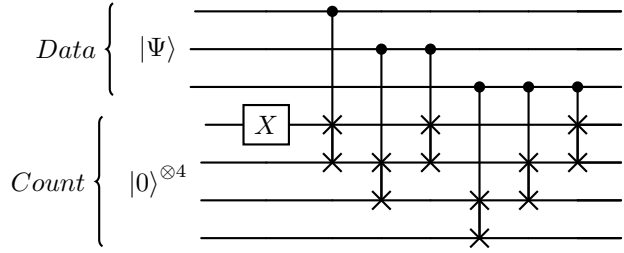
\begin{figure}[t]
  \centering
  \begin{quantikz}[row sep={0.5cm,between origins}]
    \lstick[3]{$Data$}   \midstick[3,brackets=none]{$\ket{\Psi}$}               &          & \ctrl{4} &          &          &          &          &          & \qw \\
                                                                                &          &          & \ctrl{4} & \ctrl{3} &          &          &          & \qw \\
                                                                                &          &          &          &          & \ctrl{4} & \ctrl{3} & \ctrl{2} & \qw \\
    \lstick[4]{$Count$}   \midstick[4,brackets=none]{$\ket{0}^{\otimes 4}$}     & \gate{X} & \swap{1} &          & \swap{1} &          &          & \swap{1} & \qw \\
                                                                                &          & \targX{} & \swap{1} & \targX{} &          & \swap{1} & \targX{} & \qw \\
                                                                                &          &          & \targX{} &          & \swap{1} & \targX{} &          & \qw \\
                                                                                &          &          &          &          & \targX{} &          &          & \qw
  \end{quantikz}
  \caption{Decomposition of the multi-controlled $E$ gate appearing in Figures~\ref{fig:teeth_circuit} and~\ref{fig:handle_circuit} for $n=3$. Conditioned on the data qubits, controlled-swap (Fredkin) gates are applied in cascade. This maps the number of data qubits in state $\ket{1}$ to a one-hot encoding in the $Count$ register.}
  \label{fig:one-hot_circuit}
\end{figure}

\subsubsection{Construction of the full mixing operator}

Let $\mathcal{C}_s^n$ denote the set of tuples $C=(H,T_1,\ldots,T_s)$ satisfying the comb conditions \ref{cond1} -- \ref{cond4}. The target state prepared by the unitary $U_S = U_H U_T$, realized with the circuits illustrated in Figures~\ref{fig:teeth_circuit}, \ref{fig:handle_circuit} and~\ref{fig:one-hot_circuit}, is therefore
\begin{equation}\label{eq:set_of_combs}
    \ket{F_s^n} = U_H U_T |0\rangle^n = \sum_{C \in \mathcal{C}_s^n} \kappa_C \ket{C},
\end{equation}
where $\kappa_C$ is a normalization factor that depends on the currently undefined factors $\alpha_k$ in Eq.~\eqref{eq:Dicke_superposition}. In the following section, we show how to ensure that $\kappa_C = \frac{1}{\sqrt{|\mathcal{C}_s^n|}} \; \forall \; C \in \mathcal{C}_s^n$. This leads to the  mixing unitary
\begin{equation}\label{eq:mixing_unitary}
    U_M(\beta)=e^{-i\beta \ket{F_s^n}\bra{F_s^n}} = U_H U_T (\mathds{1}-(1-e^{-i\beta})(\ket{0}\bra{0})^{\otimes n}) U_T^\dagger U_H^\dagger.
\end{equation}
Unitaries $U_S$, $U_P$ and $U_M$ can then be used by a GM-QAOA algorithm to find violated combs. With an additional multiplicative $O(n)$ time cost, one can run this GM-QAOA algorithm with different values of $s$ to search over all comb configurations.

\subsection{Uniform distribution}\label{sec:uniform_distribution}

We now explain how to choose the coefficients $\alpha_k$ in Eq.~\eqref{eq:Dicke_superposition} to ensure that $U_S$ prepares a uniform superposition of combs. First note that every distinct comb with $s$ teeth has $2s!$ representations in $\mathcal{C}_s^n$: $s!$ orderings of its teeth and two choices of handle, $H$ or $V\backslash H$. Indeed, ${\delta(H)=\delta(V\setminus H)}$, and complementing the handle preserves the comb conditions. Since this multiplicity is constant, a uniform distribution over the elements of $\mathcal{C}_s^n$ also induces a uniform distribution over distinct comb inequalities.

At any given moment during the execution of the circuit, the coefficients $\alpha_k$ must account for the number of elements in $\mathcal{C}_s^n$ that can be represented given the current partial vertex assignment and each possible subsequent assignment. Crucially, this number depends only on the number of vertices that remain available after assigning vertices to tooth $T_j$ and the number of tooth registers that remain to be filled (up to $T_s$). Let $N(c,t)$ be the number of ways to complete a partial assignment when $c$ vertices have not yet been assigned to a tooth and $t$ teeth remain. The count includes the eventual choices of the handle intersection with each remaining tooth, even though these handle qubits are prepared at a later stage of the circuit. We set $N(c,t)=0$ when $c<2t$, since every remaining tooth requires at least two vertices.

When $t=0$, all $c$ remaining vertices belong to no tooth and may independently be included in or excluded from the handle. Therefore,
\begin{equation}\label{eq:N_c_0}
    N(c,0)=2^c .
\end{equation}
For $t\geq 1$, suppose that the next tooth has cardinality $k$. There are $\binom{c}{k}$ ways to choose its vertices. Once the tooth is fixed, its intersection with the handle can be any nonempty proper subset, giving $2^k-2$ possibilities. The remaining $t-1$ teeth and handle memberships can then be completed in $N(c-k,t-1)$ ways. Hence,
\begin{equation}\label{eq:recurrence_relation}
    N(c,t) = \sum_{k=2}^{c-2(t-1)} \binom{c}{k} \bigl(2^k-2\bigr) N(c-k,t-1),
\end{equation}
where the summand in the recurrence corresponds to the number of complete encoded combs for which the next tooth has cardinality $k$. In particular, $N(n,s) = \lvert\mathcal{C}_s^n\rvert$. Consequently, when $c$ vertices and $t$ teeth remain, we configure the $D_{p,q}^r$ gates in the tooth-preparation circuit to prepare $\sum_{k=p}^q \alpha_k \ket{D_k^c}$ ($r=c$ in this context) with
\begin{equation}\label{alpha_k_tooth}
    \alpha_k = \sqrt{\frac{\binom{c}{k}\bigl(2^k-2\bigr)N(c-k,t-1)}{N(c,t)}}\:.
\end{equation}
Normalization of these amplitudes follows from Eq.~\eqref{eq:recurrence_relation}. Additionally, for a tooth of cardinality $k$, we set the $D_{p,q}^r$ gates in the handle-preparation circuit to prepare the state
\begin{equation}\label{alpha_k_handle}
    \ket{D_{1,k-1}^k} = \sum_{\ell=1}^{k-1} \sqrt{\frac{\binom{k}{\ell}}{2^k-2}}\,\lvert D_\ell^k\rangle .
\end{equation}
Given these gate choices, the state $\sum_{k=p}^q \alpha_k \ket{D_k^c} \ket{D_{1,k-1}^k}$ is shared by the register associated to this tooth and the handle register. Within a general Dicke state $\lvert D_k^n\rangle$, every computational basis state has amplitude $1/\sqrt{\binom{n}{k}}$. Consequently, for a specified tooth and a specified valid handle intersection, the combined amplitude contributed by these choices is
\begin{equation}
    \alpha_k\frac{1}{\sqrt{\binom{c}{k}}} \frac{\sqrt{\binom{k}{l}}}{\sqrt{2^k-2}}\frac{1}{\sqrt{\binom{k}{l}}} = \sqrt{\frac{N(c-k,t-1)}{N(c,t)}}.
\end{equation}

This equation gives a uniform state. To see why, let
\begin{equation}
    c_0=n,
    \qquad
    c_j=n-\sum_{i=1}^{j}\lvert T_i\rvert
    \quad (j=1,\ldots,s).
\end{equation}
After all teeth and their handle intersections have been specified, the amplitude accumulated from these choices is
\begin{equation}
    \begin{split}
        A &= \prod_{j=1}^{s}\sqrt{\frac{N(c_j,s-j)}{N(c_{j-1},s-j+1)}}\\
        &= \sqrt{\frac{N(c_1,s-1)}{N(c_{0},s)}} \sqrt{\frac{N(c_2,s-2)}{N(c_{1},s-1)}} \dots \sqrt{\frac{N(c_s,0)}{N(c_{s-1},1)}}\\
        &= \sqrt{\frac{N(c_s,0)}{N(n,s)}}.
    \end{split}
\end{equation}
The $c_s$ vertices contained in no tooth are assigned to the handle by controlled Hadamard gates, contributing $\frac{1}{\sqrt{2^{c_s}}}$ to the amplitude of each particular assignment. Therefore, setting the parameters $\alpha_k$ of the $D_{p,q}^r$ gates in the tooth-preparation circuit according to Eq.~\eqref{alpha_k_tooth} and configuring these gates in the handle-preparation circuit to prepare states of the form given in Eq.~\eqref{alpha_k_handle} gives amplitudes $\kappa_C$ of each computational basis state in the superposition prepared by $U_S$ (see Eq.~\eqref{eq:set_of_combs}) that are equal to
\begin{equation}\label{eq:amplitude_per_basis_state}
        \kappa_C = A\frac{1}{\sqrt{2^{c_s}}}
        =\sqrt{\frac{2^{c_s}}{\lvert\mathcal{C}_s^n\rvert}}\,\frac{1}{\sqrt{2^{c_s}}}
        =\frac{1}{\sqrt{\lvert\mathcal{C}_s^n\rvert}} ,
\end{equation}
where we used Eq.~\eqref{eq:N_c_0} and $N(n,s) = \lvert\mathcal{C}_s^n\rvert$.

The table of values $N(c,t)$ contains $O(ns)$ entries, and each entry requires a sum of at most $O(n)$ terms. Moreover, for $t$ remaining teeth, each of the $c$ remaining vertices can fall in one of $2t+2$ categories: for each tooth $T_j$ there are two categories ($T_j \cap H$ and $T_j \setminus H$), plus handle only and neither. Therefore,
\begin{equation}
    N(c,t)\leq (2t+2)^c\leq (n+2)^n
\end{equation}
so every table entry has $O(n\log n)$ bits. The table can therefore be computed classically in polynomial time before the quantum circuit is compiled. The values of $\alpha_k$ are then classical constants used to specify the controlled rotations.

\section{Resource analysis}\label{sec:resources}

\begin{proof}[Proof of Theorem~\ref{thm:main}]
The encoding uses one $n$-qubit register for the handle and one $n$-qubit register for each of the $s$ teeth, giving $n(s+1)$ data qubits. The $Mask$, $Dicke$, and $Count$ registers add $3n$ ancilla qubits. Therefore, $O(sn)$ logical qubits are required (error correction and routing are excluded).

We now count the number of gates required to implement the unitary $U_S$. A controlled Dicke-state superposition on at most $n$ qubits can be prepared using $O(n^2)$ gates. The first tooth requires one such block. For each later tooth, the number of available vertices can take $n-2s+1$ possible values, producing $(s-1)(n-2s+1)$ controlled Dicke blocks. Similarly, for each of the $s$ teeth, the handle-intersection circuit contains a controlled Dicke block for each of the $n-2s+1$ possible tooth cardinalities. Thus, one application of $U_S$ contains $O(sn)$ controlled Dicke blocks requiring $O(sn^3)$ gates.

The one-hot encoding circuit uses $O(n^2)$ Fredkin gates, which have constant-size decompositions into two-qubit gates. This circuit is executed for the $s-1$ later teeth and for each of the $s$ handle-intersection stages. Including computation and uncomputation, this contributes $O(sn^2)$ gates.

Each controlled cyclic permutation can be realized using $O(n)$ Fredkin gates. Controlled cyclic permutations and Fredkin gates are used $n$ times for the $s-1$ later teeth and for each of the $s$ handle-intersection stages. Therefore, these contribute $O(sn^2)$ gates. $Mask$ register updates contribute $2sn$ CNOT gates and $2n$ $X$ gates, while $n$ controlled Hadamard gates are used at the end of the circuit. The $U_S$ unitary therefore requires $O(sn^3)$ gates.

If arbitrary single-qubit rotations are not available, we can decompose the circuit into $M=O(sn^3)$ gates among the set $\{H,S,T,\mathrm{CNOT}\}$. Kliuchnikov, Maslov and Mosca~\cite{Kliuchnikov2013} give such a construction approximating any single-qubit gate to accuracy $\delta$ using $O(\log(1/\delta))$ Clifford+$T$ gates and at most two ancillary qubits. We can choose the error of each gate to be at most $\delta=\epsilon/M$. The error accumulated by applying the sequence of $M$ gates add at most linearly and the trace distance between pure states is at most their Euclidean distance \cite{Nielsen_Chuang_2010}. The trace-distance error is therefore at most $\epsilon$ and the compiled gate count is $O\!\left(M\log (M/\epsilon)\right) = O\!\left(sn^3\log(n/\epsilon)\right)$.

The tooth-preparation circuit generates a superposition over all assignments of vertices to $s$ pairwise disjoint teeth, each containing at least two vertices. For each such assignment, the handle-preparation circuit generates a superposition over all compatible handles. Thus, every element of $\mathcal{C}_s^n$ appears in the prepared superposition.

Equations~\eqref{eq:N_c_0}--\eqref{eq:amplitude_per_basis_state} show that we can ensure that every encoded comb has amplitude $1/\sqrt{N(n,s)}$. Since ${N(n,s)=|\mathcal{C}_s^n|}$, the desired uniform state is prepared. At the end of the circuit, all preparation ancillas are returned to their initial state.
\end{proof}

Regarding the phase-separation operator, each nonzero edge weight produces one commuting $ZZ$ rotation in the handle register and in each tooth register. Hence, one phase-separation layer contains $(s+1)r$ rotations, or $O(sr)$ two-qubit gates where $r=|\{e:x_e^*>0\}|$. In the dense worst case this is $O(sn^2) \subset O(n^3)$. However, this bound can be lowered as the simplex guarantees $r \leq 2n+m$ where $m$ is the number of inequalities in the system $Cx \leq d$ of Eq.~\eqref{eq:TSP_LP} ~\cite{applegate2007}.

The GM-QAOA mixing operator uses one application each of $U_S = U_H U_T$ and $U_S^\dagger = U_T^\dagger U_H^\dagger$, two layers of $X$ gates and one multi-controlled phase shift gate which can be implemented with $O(sn)$ gates~\cite{Bartschi2020}. The $U_S$ unitary therefore dominates the complexity of the mixing operator. GM-QAOA requires initializing the state with $U_S$, applying the phase-separation operator and the mixer operator once each. Therefore, executing $p$ QAOA layers requires $O(psn^3)$ gates.

Classical parameter optimization and repeated sampling incur additional costs. The present construction aims to provide a resource-friendly heuristic. When used for cut generation in branch-and-cut, parameter optimization and sampling can stop as soon as a violated comb is found. This stopping criterion may substantially reduce the practical search effort. We also note that barren plateaus can be avoided for sufficiently deep GM-QAOA circuits for a broad class of optimization problems~\cite{Tsvelikhovskiy2025}.

Nevertheless, polynomial resource requirements per circuit execution do not imply polynomial total effort to find a violated comb. Analyses of GM-QAOA established limitations on probability amplification and approximation performance~\cite{Bridi2024,Xie2025}. In particular, if a fraction $\rho$ of the encoded combs are violated, the probability of measuring a violated comb after $p$ layers is at most $\min\{1,(2p+1)^2\rho\}$. If $\rho$ is exponentially small, achieving a fixed success probability can therefore require exponentially many layers or repeated executions. However, it may also be the case that for some nodes of the branch-and-cut tree, $\rho$ is not so small and many combs represent violated inequalities that can be added to the LP relaxation. A precise study of $\rho$ for practical TSP instances, and as such of the expected performance of GM-QAOA for the separation of comb inequalities, represents the natural next step of this study.

\section{Conclusion}\label{sec:conclusion}

We have presented a quantum heuristic separation algorithm for the class of comb inequalities in the TSP, and shown how it can be included in a hybrid branch-and-cut framework. Its running time is polynomial when polynomial budgets are prescribed for the number of QAOA rounds, parameter optimization steps and samples gathered. The algorithm prepares a quantum superposition of all possible combs in a quantum register and then runs the Grover-Mixer Quantum Alternating Operator Ansatz algorithm. We showed how to build the state-preparation unitary providing the superposition of combs and gave an explicit description of the cost Hamiltonian to be used in the phase-separator unitary. The procedure can be applied to every odd integer $s$ satisfying $3\leq s\leq\lfloor n/2\rfloor$. There are only $O(n)$ such values, so allowing the number of teeth to vary preserves polynomial running time and makes the full family of comb inequalities accessible to the search. The performance of our approach relative to classical methods remains to be assessed numerically and experimentally.

A practical motivation for our approach is its integration into a classical branch-and-cut solver. Applying GM-QAOA directly to the TSP, as considered by B\"artschi and Eidenbenz~\cite{Bartschi2020}, searches for feasible tours and therefore provides upper bounds on the optimal tour length. Our approach instead searches for inequalities that can strengthen the LP relaxation and improve lower bounds, while the tour itself is provided by the classical linear program solver. Each sampled candidate inequality can be checked classically for validity and violation before it is added to the relaxation. If the quantum routine fails to find a violated comb, the existing LP bounds remain valid, although an opportunity to strengthen them may be missed. Provided that the classical branching, feasibility checks, and pruning rules are retained, the correctness of the solver does not depend on the success of the quantum heuristic.

As shown by Gr\"{o}tschel and Padberg \cite{Grotschel1979a}, the number of combs is given by
\begin{equation}\label{eq:number_combs}
\nu^C(n) = \sum_{q=3}^{n-3} \frac{1}{2}\binom{n}{q}
           \sum_{j=3}^{n-q} \binom{n-q}{j}
           \sum_{\substack{k=3 \\ k \text{ odd}}}^{\min(j,q)}
           A_k^j \sum_{p=k}^{q} \frac{1}{k!} \binom{q}{p} A_k^p,
\end{equation}
where
\begin{equation}
A_k^p = \sum_{j=0}^{k} (-1)^j \binom{k}{j} (k-j)^p.
\end{equation}
Here, combs differing only by tooth permutations or handle complementation are counted as the same comb. Our encoding retains these equivalent representations. For fixed $s$, the number of basis states is $N(n,s)$, whereas the corresponding contribution to $\nu^C(n)$ is $N(n,s)/(2s!)$. On the other hand, the number of undirected tours in the TSP is given by $\frac{(n-1)!}{2}$. As illustrated in Figure~\ref{fig:ratio_tours_combs}, the ratio of the number of tours to the number of distinct combs grows exponentially with the number of cities. Consequently, an encoding assigning a single basis state to encode all $2s!$ representations of each distinct comb would yield an exponentially smaller search space than an encoding assigning one state to each undirected tour. In our encoding, the multiplicity of comb representations is the same for each distinct comb for a fixed $s$, so the comparison remains meaningful in terms of the number of inequivalent candidates. This observation of a drastically reduced search space further motivates the approach we propose compared to applying QAOA directly to the TSP as considered in~\cite{Bartschi2020}.

\begin{figure}
  \centering
  \includegraphics[width=\textwidth]{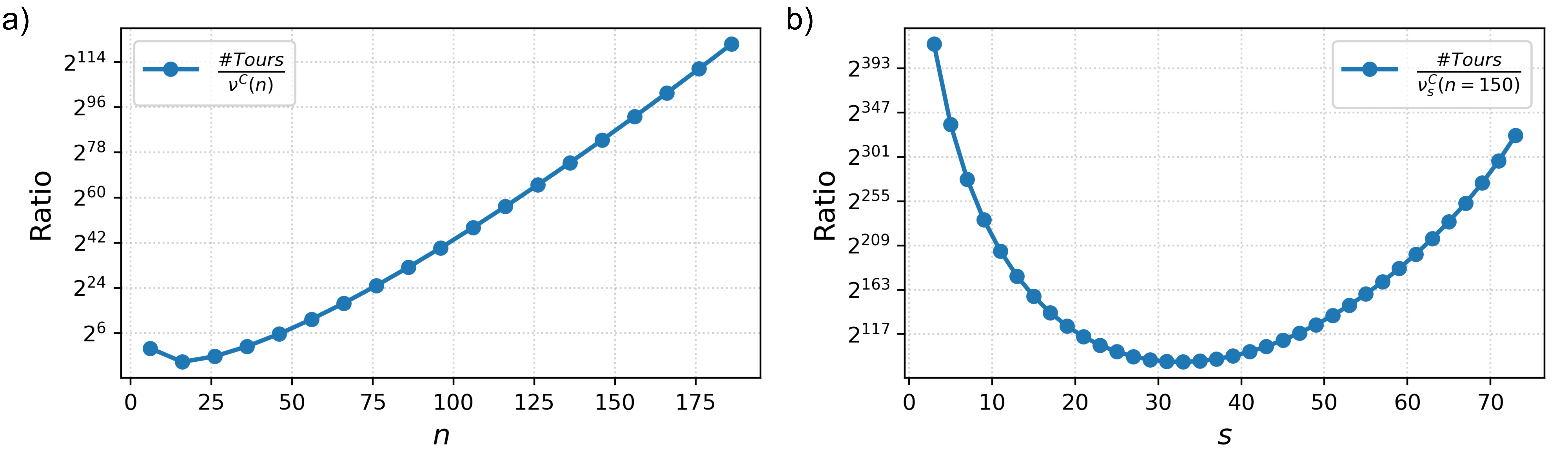}
  \caption{Panel a) shows the ratio of the number of undirected TSP tours to the number of combs $\nu^C$ given by Eq.~\eqref{eq:number_combs}, as a function of the number of vertices $n$. Panel b) shows the corresponding ratio for combs with exactly $s$ teeth $\nu_s^C$, as a function of $s$, with $n=150$. In both panels, combs differing only by tooth permutations or handle complementation are counted as the same comb.}
  \label{fig:ratio_tours_combs}
\end{figure}

For large, practically relevant TSP instances, the substantial qubit and gate requirements of the proposed algorithm make classical simulation and execution on current quantum hardware impractical. Evaluation at this scale should be pursued as computational resources and quantum hardware improve. This work also suggests directions for extending quantum heuristic separation to other families of valid inequalities, including clique-tree and path inequalities. Such extensions would require encodings of their defining combinatorial structures, efficient preparation of feasible states, and corresponding phase-separation operators. Another direction is to investigate decoded quantum interferometry~\cite{Jordan2025} as an alternative to variational optimization. Applying this framework to comb separation would require a suitable formulation of the objective and feasibility constraints together with an efficient decoding procedure. Determining whether such a formulation can exploit the structure of combs remains an open question.

We hope that this work encourages further investigation of quantum subroutines for cutting-plane generation and, more broadly, of hybrid algorithms in which quantum resources contribute to the performance of classical optimization methods.

\printbibliography

\end{document}